\documentclass[pdflatex,sn-mathphys-num]{sn-jnl}

\usepackage{graphicx}%
\usepackage{multirow}%
\usepackage{amsmath,amssymb,amsfonts}%
\usepackage{amsthm}%
\usepackage{mathrsfs}%
\usepackage[title]{appendix}%
\usepackage{xcolor}%
\usepackage{textcomp}%
\usepackage{booktabs}%
\usepackage{mathrsfs}
\usepackage[title]{appendix}%
\usepackage{xcolor}%
\usepackage{textcomp}%
\usepackage{manyfoot}%
\usepackage{booktabs}%

\usepackage{dcolumn}
\usepackage{bm}

\usepackage{amsthm}
\usepackage[mathcal]{euscript}
\usepackage{tikz}
\usetikzlibrary{trees,arrows,decorations.pathreplacing}

\usepackage{epsdice}
\usepackage{url}
\usepackage{hyperref}
\usepackage{color}
\definecolor{refcolor}{RGB}{0,0,190}
\hypersetup{
    colorlinks,
    citecolor=refcolor,
    filecolor=refcolor,
    linkcolor=refcolor,
    urlcolor=refcolor
}

\usepackage{booktabs}

\usepackage{bookmark}
\bookmarksetup{
  numbered, 
  open,
}

\newtheorem{theorem}{Theorem}
\newtheorem{proposition}{Proposition}

\theoremstyle{remark}
\newtheorem{remark}{Remark}

\theoremstyle{definition}
\newtheorem{observation}{Observation}
\newtheorem{example}{Example}

\newtheorem{definition}{Definition}

\newtheorem{principle}{Principle}

\newtheorem{condition}{Condition}
\newtheorem{difficulty}{Difficulty}

\theoremstyle{definition}
\newtheorem{defCustom}{Definition}
\renewcommand{\thedefCustom}{\arabic{definition}}
\makeatletter
\newcommand{\setdefCustomtag}[1]{
  \let\oldthedefCustom\thedefCustom
  \renewcommand{\thedefCustom}{#1}
  \g@addto@macro\enddefCustom{
    \global\let\thedefCustom\oldthedefCustom}
  }
\makeatother

\newtheorem{thesisCustom}{Thesis}
\renewcommand{\thethesisCustom}{\arabic{thesis}}
\makeatletter
\newcommand{\setthesisCustomtag}[1]{
  \let\oldthethesisCustom\thethesisCustom
  \renewcommand{\thethesisCustom}{#1}
  \g@addto@macro\endthesisCustom{
    \global\let\thethesisCustom\oldthethesisCustom}
  }
\makeatother

\newtheorem{assumption}{Assumption}
\renewcommand{\theassumption}{\arabic{assumption}}
\makeatletter
\newcommand{\setassumptiontag}[1]{
  \let\oldtheassumption\theassumption
  \renewcommand{\theassumption}{#1}
  \g@addto@macro\endassumption{
    \global\let\theassumption\oldtheassumption}
  }
\makeatother

\newtheorem{claim}{Claim}
\renewcommand{\theclaim}{\arabic{claim}}
\makeatletter
\newcommand{\setclaimtag}[1]{
  \let\oldtheclaim\theclaim
  \renewcommand{\theclaim}{#1}
  \g@addto@macro\endclaim{
    \global\let\theclaim\oldtheclaim}
  }
\makeatother

\theoremstyle{remark}
\newtheorem{pointItem}{}
\renewcommand{\thepointItem}{\quad\arabic{pointItem}}
\makeatletter
\newcommand{\setpointItemtag}[1]{
  \let\oldthepointItem\thepointItem
  \renewcommand{\thepointItem}{#1}
  \g@addto@macro\endpointItem{
    \global\let\thepointItem\oldthepointItem}
  }
\makeatother

\def\({\left(}
\def\){\right)}

\newcommand{\tn}{\textnormal}

\newcommand{\boplus}{\textstyle{\bigoplus}}

\newcommand{\hilbert}{\mathcal{H}}

\newcommand{\mc}[1]{\mathcal{#1}}

\newcommand{\mf}[1]{\mathfrak{#1}}

\newcommand{\wh}[1]{\widehat{#1}}

\newcommand{\wt}[1]{\widetilde{#1}}

\newcommand{\R}{\mathbb{R}}
\newcommand{\C}{\mathbb{C}}

\newcommand{\N}{\mathbb{N}}

\newcommand{\abs}[1]{\left\lvert#1\right\rvert}

\newcommand{\de}{\operatorname{d}}

\newcommand{\tr}{\operatorname{tr}}

\newcommand{\ie}{\textit{i.e.}\ }
\newcommand{\vs}{\textit{vs.}\ }
\newcommand{\eg}{\textit{e.g.}\ }
\newcommand{\cf}{\textit{cf.}\ }
\newcommand{\etc}{\textit{etc}}

\newcommand{\U}{\tn{U}}

\newcommand{\schrod}{Schr\"odinger}

\newcommand{\bra}[1]{\langle#1|}
\newcommand{\ket}[1]{|#1\rangle}

\newcommand{\braket}[2]{\langle#1|#2\rangle}

\newcommand{\pobs}[1]{\mathsf{#1}}
\newcommand{\obs}[1]{\wh{\pobs{#1}}}

\newcommand{\x}{\mathbf{x}}
\newcommand{\y}{\mathbf{y}}

\newcommand{\n}{\mathbf{n}}

\def\sref #1{\S\ref{#1}}
\def\eqref #1{(\ref{#1})}

\newcommand{\image}[3]{
\vspace{-0.1in}
\begin{center}
\begin{figure}[!ht]
\includegraphics[width=#2\textwidth]{#1}
\caption{\small{\label{#1}#3}}
\end{figure}
\end{center}
\vspace{-0.4in}
}

\title[Born rule as classical probability]{Born rule: quantum probability as classical probability}

\author*{\fnm{Ovidiu Cristinel} \sur{Stoica}}\email{cristi.stoica@theory.nipne.ro,holotronix@gmail.com}

\affil*{\orgdiv{Dept. of Theoretical Physics}, \orgname{NIPNE---HH, Bucharest}, \city{Bucharest}, \country{Romania}, ORCID: \href{https://orcid.org/0000-0002-2765-1562}{0000-0002-2765-1562}}

\abstract{I provide a simple derivation of the Born rule as giving a classical probability, that is, the ratio of the measure of favorable states of the system to the measure of its total possible states.

In classical systems, the probability is due to the fact that the same macrostate can be realized in different ways as a microstate. Despite the radical differences between quantum and classical systems, I show that the same can be applied to quantum systems, and the result is the Born rule.

This works only if the basis is continuous (an eigenbasis of observables with continuous spectra), but all known physically realistic measurements involve a continuous basis (the position basis).

The continuous basis is not unique, and for subsystems it depends on the observable.

But for the entire universe, there are continuous bases that give the Born rule for all measurements, because all measurements reduce to distinguishing macroscopic pointer states, and macroscopic observations commute. This allows for the possibility of a unique ontic basis for the entire universe. 

In the wavefunctional formulation, the basis can be chosen to consist of classical field configurations, and the coefficients $\Psi[\phi]$ can be made real by absorbing them into a global U(1) gauge.

For the many-worlds interpretation, this result gives the Born rule from micro-branch counting.}


\begin{document}

\maketitle

\section{Introduction}
\label{s:intro}

In quantum mechanics, the Born rule prescribes that the probability that the result of a quantum measurement is the eigenvalue $\lambda_j$ of the observable is
\begin{equation}
\label{eq:born_rule}
\tn{Prob}(\lambda_j)=\bra{\psi}\obs{P}_j\ket{\psi},
\end{equation}
where the unit vector $\ket{\psi}$ represents the state of the observed system right before the measurement, and $\obs{P}_j$ is the projector on the eigenspace corresponding to $\lambda_j$.

The \emph{Projection Postulate} states that $\ket{\psi}$ projects onto one of the eigenspaces $\obs{P}_j$ with a probability given by \eqref{eq:born_rule}.

von Neumann expressed already in $1927$ the desirability of having a derivation of the Born rule ``from empirical facts or fundamental probability-theoretic assumptions, \emph{i.e.}, an inductive justification'' \cite{vonNeumann1955MathFoundationsQM}.
Gleason's theorem shows that any countably additive probability measure on closed subspaces of a Hilbert space $\hilbert$, $\dim\hilbert>2$, has the form $\tr(\obs{P}\wh{\rho})$, where $\obs{P}$ is the projector on the subspace and $\wh{\rho}$ is a density operator \cite{GleasonTheorem1957}. If the state is represented by $\wh{\rho}$, this can be interpreted as the Born rule.
Gleason's theorem is very important, in showing that if there is a probability rule, it should have the form of the Born rule.
But it does not say that the density operator of the observed system is the same $\wh{\rho}$, how the probabilities arise in the first place, and what they are about \cite{Earman2022TheStatusOfTheBornRuleAndTheRoleOfGleasonsTheoremAndItsGeneralizations}. For example, it is unable to convert the amplitudes of the branches in the many-worlds interpretation (MWI) \cite{Everett1957RelativeStateFormulationOfQuantumMechanics,deWittGraham1973ManyWorldsInterpretationOfQuantumMechanics,Wallace2012TheEmergentMultiverseQuantumTheoryEverettInterpretation,SEP-Vaidman2021MWI} into actual probabilities. For this reason, the search for a proof of the Born rule continues.

There are numerous proposals to derive the Born rule.
Earlier attempts to derive it from more basic principles include \cite{Finkelstein1963LogicOfQuantumPhysics_Born_rule_derivation}, \cite{Hartle1968QuantumMechanicsOfIndividualSystems}, \cite{FarhiEtal1989HowProbabilityArisesInQuantumMechanics} \etc. Such approaches based on a frequency operator were accused of circularity \cite{Cassinello1996OnTheProbabilisticPostulateOfQuantumMechanics,CavesSchack2005PropertiesOfTheFrequencyOperatorDoNotImplyTheQuantumProbabilityPostulate}.
Other proposals, in relation to the MWI, are based on imposing conditions like additivity to the probability measure \cite{Everett1957RelativeStateFormulationOfQuantumMechanics}, many-minds
\cite{AlbertLower1988InterpretingMWI_ManyMinds}, decision theory \cite{Deutsch1999QuantumTheoryOfProbabilityAndDecision,Wallace2002QuantumProbabilitiesAndDecisionRevisited,Saunders2004BornRuleFromOperationalAssumptions} (accused of circularity in \cite{Baker2007MeasurementOutcomesAndProbabilityInEverettianQuantumMechanics,BarnumEtal2000QuantumProbabilityFromDecisionTheory}), envariance \cite{Zurek2005ProbabilitiesFromEntanglement} (accused of circularity in \cite{SchlosshauerFine2003OnZureksDerivationOfTheBornRule}), measure of existence \cite{Vaidman2012ProbabilityInMWI}, branch counting based on refinements of the branching structure \cite{Saunders2021BranchCountingInTheEverettInterpretationOfQuantumMechanics} \etc. For a review see \cite{Vaidman2020DerivationsOfTheBornRule}. 

In this article, I investigate the possibility of obtaining probabilities that are very similar to the classical ones.
As in classical physics, what we observe are not actually the states (also named \emph{microstates} in this context), but the macrostates.
If each macrostate can be realized in different ways as a microstate, probabilities can arise from the relative count, or rather the relative measure (since the basis is uncountable), of the states underlying each macrostate, just like in the standard understanding of probabilities.

In Section \sref{s:probabilities} I argue that, contrary to the common view on quantum mechanics, an ``ontic'' or ``classical'' basis for the entire universe is possible, allowing for classical-like probabilities in quantum mechanics.

In Sec. \sref{s:born} I prove the main result, that probability density can be understood as a distribution of ``ontic'' or ``classical'' states.

In Sec. \sref{s:interpretation} I discuss the physical interpretation of this derivation of the Born rule, how it makes possible the existence of a ``classical'' or ontic basis for the entire universe, how complex numbers appear, and how this yields probabilities in the many-worlds interpretation.

There exist already several interesting proposals to understand the Born rule classically \cite{Bell2004TheMeasurementTheoryOfEverettAndDeBrogliesPilotWave,Tipler2006WhatAboutQuantumTheoryBayesAndTheBornInterpretation,Poirier2010BohmianMechanicsWithoutPilotWaves,SchiffPoirier2012CommunicationQuantumMechanicsWithoutWavefunctions,HallDeckertWiseman2014QuantumPhenomenaModeledByInteractionsBetweenManyClassicalWorlds,Bostrom2015QuantumMechanicsAsADeterministicTheoryOfAContinuumOfWorlds,Sebens2015QuantumMechanicsAsClassicalPhysics,Arve2020EverettsMissingPostulateAndTheBornRule,Tappenden2023SetTheoryAndManyWorlds}, and they will be compared with the one proposed here in Section \ref{s:discussion}. A central difference is that the proposal presented here can get rid completely of the wavefunction, incorporating it fully in the classical worlds, while the content of the theory remains the same as in the wavefunctional formulation of quantum theory. The final section also argues for the advantages of this way of understanding probabilities classically without having to assume counterfactual worlds that are possible but non-physical.

\section{Classical {\vs} quantum probabilities}
\label{s:probabilities}

In this Section we investigate the differences and similarities between classical probabilities and probabilities in quantum mechanics. This will provide the physical justification to interpret probabilities in a quantum universe, based on the proof given in Sec. \sref{s:born}.

Let us recall what classical probabilities are.

\begin{definition}[Classical probabilities]
\label{def:probability-space}
A \emph{probability space} $(\Omega,\mc{F},P)$ consists of a \emph{sample space} $\Omega$ (the set of possible \emph{outcomes}), an \emph{event space} $\mc{F}$ (a $\sigma$-algebra of subsets of $\Omega$), and a \emph{probability function} $P:\mc{F}\to [0,1]$.
\end{definition}

Even in a classical and deterministic universe, where a ``Laplace's demon'' who knows the physical state in full detail should be able to assign to each event a probability equal to either $0$ or $1$, the evolution of a system can appear to be unpredictable at the macro-level. The reason is that the state cannot be known exactly. We can only access the macrostates. The macrostates correspond to events, and the microstates to outcomes, \cf Definition \ref{def:probability-space}.

\begin{observation}
\label{obs:agents}
Nontrivial probabilities exist for agents who lack complete information about the microstate.
\end{observation}

For example, the classical probability that throwing a pair of dice results in the events
\begin{equation}
\label{eq:dice_event}
\mbox{\LARGE$
\vcenter{\hbox{\epsdice{6}}}+\vcenter{\hbox{\epsdice{6}}}
$}
\end{equation}
is given by the probability measure of the set of microstates (outcomes) that realize the macrostate in which the event is \eqref{eq:dice_event} divided by the total probability measure.

Therefore, classical probabilities satisfy

\begin{condition}[Probability]
\label{cond:probability}
The probability is \emph{the ratio of the measure of favorable outcomes to the total measure of possible outcomes}.
\end{condition}

In classical physics, Condition \ref{cond:probability} seems to make sense because the universe is in a unique state at any time. But in quantum mechanics, a system can be in a superposition of multiple states that coexist in parallel.

\begin{remark}
\label{rem:classical-probability}
The classical interpretation of probabilities has its own problems, in particular in the continuum case, mainly its reliance on the Principle of Indifference \cite{sep-probability-interpret,Wallace2012TheEmergentMultiverseQuantumTheoryEverettInterpretation}. 
Note that the very word ``classical'' is ambiguous in this context, as it can be understood in two ways.
On one hand ``classical'' refers to the classical interpretation of probabilities, in contrast with alternatives like the frequentist or the Bayesian approaches.
On the other hand, it refers to probabilities as known in a classical world, in contrast with quantum probabilities, regardless of the question of their interpretation mentioned above and discussed in \cite{sep-probability-interpret}.
Quantum mechanics adds a new layer of problems to probabilities, because a system can be in a superposition of multiple states that coexist in parallel, and there seem to exist multiple ways to see a quantum state as multiple classical states.

This article does not aim to vindicate classical probability compared to alternatives like the frequentist approach and the Bayesian approach. In the view of this author they complement each other, and they all come in classical and quantum flavors, although there are very strong reasons to adopt the Bayesian interpretation.
The aim of this article is rather to explain quantum probabilities as probabilities in a classical world.
Condition \ref{cond:probability} mentions the classical interpretation simply because it is more commonly used and simple to understand mathematically as in Definition \ref{def:probability-space}. Once Condition \ref{cond:probability} is met, one can apply their own preferred framework. We will return to this in Section \sref{s:discussion}.
\end{remark}

\begin{difficulty}
\label{difficulty:parallel}
Unlike classical states, quantum states seem to be able to coexist in parallel, in a superposition, as shown by interference experiments. Moreover, the states in superposition have completely different meanings in different bases.
\end{difficulty}

From Difficulty \ref{difficulty:parallel}, another difficulty follows.
In a quantum universe, the central difference seems to be that there are multiple ways in which the macrostate of a subsystem can be realized as microstates, each depending on the experimental settings. For example, the spin of a particle can be interpreted as consisting of definite possible spins $\ket{\uparrow}_z$ and $\ket{\downarrow}_z$ if the spin is measured along the axis $z$, but not if it is measured along another axis. In classical mechanics, the possible results of the measurement are considered to be independent of the measurement settings, provided that the effect  of the observation on the observed system can be made arbitrarily small.

The main difficulty in the applicability of classical probabilities in quantum mechanics is therefore

\begin{difficulty}
\label{difficulty:nonunique-probability-space}
Unlike classical systems, the probability space (the set of microstates) seems to depend on the measurement settings.
\end{difficulty}

But in reality the situation is not necessarily as described in Difficulty \ref{difficulty:nonunique-probability-space}.
To see this, consider an example.

The measurement of the spin of a spin-$1/2$ particle along the axis $z$ results in the possible states
\begin{equation}
\label{eq:spin_pointer_z}
\begin{cases}
\ket{\uparrow}_z\ket{\tn{up}}_z \\
\ket{\downarrow}_z\ket{\tn{down}}_z, \\
\end{cases}
\end{equation}
where $\ket{\uparrow}_z$ and $\ket{\downarrow}_z$ are the spin states of the particle along the axis $z$, and $\ket{\tn{up}}_z$ and $\ket{\tn{down}}_z$ are the corresponding states of the pointer of the measuring device.
A spin measurement along the axis $x$ leads to a different decomposition,
\begin{equation}
\label{eq:spin_pointer_x}
\begin{cases}
\frac{1}{\sqrt{2}}\(\ket{\uparrow}_z+\ket{\downarrow}_z\)\ket{\tn{up}}_x \\
\frac{1}{\sqrt{2}}\(\ket{\uparrow}_z-\ket{\downarrow}_z\)\ket{\tn{down}}_x. \\
\end{cases}
\end{equation}

The macrostates corresponding to the results of the spin measurements, including those along distinct axes, are orthogonal, because we can tell by examining the experimental setup along what direction in space the spin is measured, and what result is obtained.
Therefore, the microstates from eq. \eqref{eq:spin_pointer_z} are orthogonal to those from eq. \eqref{eq:spin_pointer_x}, even though the states of the observed system $\ket{\uparrow}_z$ and $\ket{\downarrow}_z$  are not orthogonal to $1/\sqrt{2}\(\ket{\uparrow}_z\pm\ket{\downarrow}_z\)$.

In general, every quantum measurement ultimately becomes a direct observation of the macrostate of the measuring device. So every measurement reduces to distinguishing macrostates. Macrostates are distinguished by macro-observables, and all macro-observables commute.

\emph{Macrostates} are represented by subspaces of the form $\obs{P}_{\alpha}\hilbert$,
where $(\obs{P}_{\alpha})_{\alpha\in\mc{A}}$ is a set of commuting projectors on $\hilbert$, so that $[\obs{P}_{\alpha},\obs{P}_{\beta}]=0$ for any $\alpha\neq\beta\in\mc{A}$, and $\boplus_{\alpha\in\mc{A}}\obs{P}_{\alpha}\hilbert=\hilbert$.
This claim is empirically adequate, as illustrated by the example of spin measurements.
This position is adopted for example in decohering histories approaches \cite{GellMannHartle1990DecoheringHistories} and in the MWI \cite{Wallace2012TheEmergentMultiverseQuantumTheoryEverettInterpretation}.

\begin{observation}
\label{obs:macrostates}
We never observe the microstate, only the macrostates.
\end{observation}

Since ultimately every measurement translates into an observation represented by the macro projectors, there is a universal basis for all measurements, which diagonalizes all macro projectors.
This is not true for subsystems, but this is irrelevant, since any measurement is ultimately one of macrostates.

\begin{observation}
\label{obs:ontic-basis}
For the entire universe, whose states are represented by vectors in a Hilbert space $\hilbert$, there is a universal basis
\begin{equation}
\label{eq:ontic_basis}
\(\ket{\phi}\)_{\phi\in\mc{C}}
\end{equation}
compatible with the macrostates. In general, more such bases exist, otherwise the macrostates would coincide with the microstates. 
\end{observation}

It may seem too much to account for states of the entire universe just to explain the probabilities of the measurement of a single particle.
But in fact we always do this, because the observed particle can be entangled with any other system in the universe, so when we measure it, we measure in fact the state of the universe.
The usual separation between the observed system and the rest of the universe that enters in our theoretical description is an idealization that may make us not see the forest for the trees.
Therefore,

\begin{observation}
\label{obs:single_state}
The state of the universe is a single quantum state, not the composition of the states of the subsystems.
\end{observation}

Now let us return to Difficulty \ref{difficulty:parallel}.
Despite mentioning a superposition of ``parallel microstates'', Difficulty \ref{difficulty:parallel} does not apply exclusively to the MWI, but also to standard quantum mechanics, because any state can be (and in general it is) in a superposition before the Projection Postulate is invoked. And even after that it remains in a superposition of microstates, as a succession of compatible measurements confirm it.

In a classical universe, when macrostates can be realized in more possible ways as microstates, probabilities arise from the fact that the agent does not know the microstate.
But if there would be more classical worlds, distributed by the same probability distribution as in the case when there is a unique world whose microstate is unknown, would there be any difference to the agent?
Would Condition \ref{cond:probability} be invalidated in such a universe?
To have an intuition of this question, consider an agent inhabiting a classical world, who tosses a coin, using a machine, so that the result is unpredictable to the agent.
If the result is heads, the agent finds out that the microstate of the world was one leading to heads.
Does this exclude the possibility of one or more classical universes in which the result was tail?
It definitely doesn't. So, while the agent may have learned something from tossing the coin, she didn't learn whether her classical world is unique or there are other classical worlds inaccessible to her.
The point here is not whether or not she should accept the existence of parallel classical worlds inaccessible to her, the point of the argument is that, just like a classical world is possible, more parallel classical worlds are possible too, and no experiment can tell which is the case, not even probabilistic experiments like tossing coins.
Both situations accommodate probabilities, the one-classical-world scenario allows probabilities coming from the ignorance of the microstate, and the many-classical-worlds scenario comes from the ignorance of the agent whether it is an instance of herself in one classical world or another, that is, her self-location.
This leads to the following:

\begin{observation}[Equivalence]
\label{obs:equivalence}
Probabilities for a given macrostate are independent of whether there is a unique microstate and the agent does not know it (\emph{epistemic probability}), or if it describes more coexisting microstates, and the agent does not know in which of these microstates it exists (\emph{self-location probability}).
\end{observation}

We can modify the example of a classical universe with many microstates evolving in parallel, by allowing the evolution to ``scramble'' the microstates, in the sense that the dynamics cannot be applied to individual microstates, but only to sets of such states, so that the measure describing how are they distributed is conserved. For example, a single microstate may evolve into multiple microstates, so that the density is preserved.
If we compare such a classical universe with one with the same state space, but so that only one microstate exists, and it evolves indeterministically according to the same measure, Observation \ref{obs:equivalence} would still apply.

An implicit assumption underlying Observation \ref{obs:equivalence} is that an agent or observer \emph{supervenes} (in the sense that its states depends) on a single microstate, even if there are more microstates coexisting in parallel.
But if there can be more parallel microstates, an additional condition is needed:
\begin{condition}[Correspondence]
\label{cond:correspondence}
If there are more parallel microstates, and at a given time different instances of an agent exist in more of them, each instance of the agent supervenes on only one of these microstates.
\end{condition}

In other words, the physical microstates should be able to support ontologically the existence of agents or observers, so that their experience of probabilities depends on their ignorance of the microstate.

Despite Difficulty \ref{difficulty:nonunique-probability-space}, the existence of a basis as in Observation \ref{obs:ontic-basis} will turn out to make it possible for quantum mechanics to satisfy Condition \ref{cond:correspondence}. This requires
\begin{principle}
\label{pp:correspondence}
In the quantum universe, there is a basis as in Observation \ref{obs:ontic-basis}, so that all instances of an agent can be realized only in microstates from that basis.
We will call it \emph{ontic basis} and its elements \emph{ontic states}.
\end{principle}

This, as we will see, allows probabilities in quantum mechanics to be just like classical probabilities, provided that we apply them to the entire universe.

Another difficulty in quantum mechanics is that the state is found, after wavefunction collapse or decoherence, to be $\obs{P}_{\alpha}\ket{\psi}/|\obs{P}_{\alpha}\ket{\psi}|$. The Born rule prescribes the probability that the state vector becomes $\obs{P}_{\alpha}\ket{\psi}/|\obs{P}_{\alpha}\ket{\psi}|$ is given by $\bra{\psi}\obs{P}_{\alpha}\ket{\psi}$, as in \eqref{eq:born_rule}. In the ontic basis \eqref{eq:ontic_basis}, $\ket{\psi}$ is decomposed as a linear combination with distinct coefficients (amplitudes) $\braket{\phi}{\psi}$, so
\begin{equation}
\label{eq:proj_ontic_decomp}
\tn{Prob}(\alpha)=\bra{\psi}\obs{P}_{\alpha}\ket{\psi}=\int_{\mc{C}_{\alpha}}\abs{\braket{\phi}{\psi}}^2d\mu(\phi),
\end{equation}
where $\mc{C}_{\alpha}=\{\phi\in\mc{C}|\ket{\phi}\in\obs{P}_{\alpha}\hilbert\}$ and $\mu$ is the measure on $\mc{C}$. While $\abs{\braket{\phi}{\psi}}^2$ gives a measure, this is not sufficient to interpret it probabilistically:
\begin{difficulty}
\label{difficulty:measure}
Just because it is a measure, it does not mean that it is a probability.
\end{difficulty}

In Sec. \sref{s:born} we will see that Principle \ref{pp:correspondence} allows Condition \ref{cond:correspondence} to be satisfied in quantum mechanics, and Difficulties \ref{difficulty:parallel}, \ref{difficulty:nonunique-probability-space}, and \ref{difficulty:measure} to be avoided, so by Observation \ref{obs:equivalence}, Condition \ref{cond:probability} is satisfied too.

\section{Derivation of the Born rule}
\label{s:born}

Before proving the main result, let us motivate it.

Consider a state vector of the form
\begin{equation}
\label{eq:finite_case_psi}
\ket{\psi}=\frac{1}{\sqrt{n}}\sum_{k=1}^n\ket{\phi_k}.
\end{equation}
where $(\ket{\phi_k})_{k\in\{1,\ldots,n\}}$ are orthonormal vectors from $\hilbert$.
Then, if every $\ket{\phi_k}$ is an eigenvector of the operator $\obs{A}$ representing the observable, the Born rule simply coincides with counting basis states:
\begin{equation}
\label{eq:finite_case_proof}
\begin{aligned}
\bra{\psi}\obs{P}_j\ket{\psi}
&=\frac{1}{n}\(\sum_{k=1}^n\bra{\phi_k}\)\(\obs{P}_j\sum_{k=1}^n\ket{\phi_k}\)\\
&=\frac{1}{n}\sum_{\ket{\phi_k}\in\obs{P}_j\hilbert}\braket{\phi_k}{\phi_k} = \frac{n_j}{n},
\end{aligned}
\end{equation}
where $\obs{P}_j$ is the projector of the eigenspace corresponding to the eigenvalue $\lambda_j$, and $n_j$ is the number of basis vectors $\ket{\phi_k}$ that are eigenvectors for $\lambda_j$.

In the MWI, this seems to be the only situation when ``naive branch counting'' coincides with the Born rule. It is used as a starting point for arguments that the Born rule is present in the MWI for ``less naive'' counting rules such as \cite{Deutsch1999QuantumTheoryOfProbabilityAndDecision,Wallace2002QuantumProbabilitiesAndDecisionRevisited,Vaidman2012ProbabilityInMWI,Saunders2021BranchCountingInTheEverettInterpretationOfQuantumMechanics}.

Eq. \eqref{eq:finite_case_proof} satisfies Condition \ref{cond:probability}, but unfortunately it seems to work only for very special state vectors.
We cannot make it work for all vectors in a finite-dimensional Hilbert space, since the eigenbasis vectors either have to contribute to eq. \eqref{eq:finite_case_psi} with the same absolute value $1/\sqrt{n}$ of the coefficient $\braket{\phi_k}{\psi}$, or to be absent.

Interestingly, as if by magic, the idea works in the continuous case without problems, because the basis vectors can be distributed with nonuniform density, making it possible for the continuous version of eq. \eqref{eq:finite_case_psi} to apply to any state vector.

To see this, we have to show that we can take a kind of continuous limit of eq. \eqref{eq:finite_case_psi} while keeping the vectors from \eqref{eq:finite_case_psi} distinct. It is important to keep them distinct, because the finite case proof from eq. \eqref{eq:finite_case_proof} only works when they are distinct. While this is a very severe limitation in finite-dimensional Hilbert spaces, it works in the infinite-dimensional case, in a continuous basis.

Let $\mc{C}$ be a topological manifold with a measure $\mu$ on its $\sigma$-algebra, and $\hilbert:=L^2(\mc{C},\mu,\C)$ the Hilbert space of square-integrable complex functions on $\mc{C}$.
Let $(\ket{\phi})_{\phi\in\mc{C}}$ be an orthogonal basis of $\hilbert$, so that
\begin{equation}
\label{eq:basis}
\int_{\mc{C}}\braket{\phi}{\phi'} \psi(\phi')d\mu(\phi') = \psi(\phi)
\end{equation}
for any compact-supported continuous function $\psi\in \hilbert$.

Suppose that the projectors $(\obs{P}_{\alpha})_{\alpha\in\mc{A}}$ associated to the possible results of the measurements are \emph{compatible} with the basis $(\ket{\phi})_{\phi\in\mc{C}}$, that is, each subspace $\obs{P}_{\alpha}\hilbert$ is spanned by the vectors $\ket{\phi}$, $\phi\in\mc{C}_{\alpha}$, for some $\mu$-measurable set $\mc{C}_{\alpha}\subseteq\mc{C}$.

In the following, by \emph{continuous basis} we understand one that is a common eigenbasis for a complete set of commuting observables with continuous spectra.

\begin{theorem}
\label{thm:born_rule_counting}
Let $(\ket{\phi})_{\phi\in\mc{C}}$ be a continuous basis compatible with the projectors associated to the possible results.
Then, for any state vector $\ket{\psi}$, the continuous limit of counting basis vectors gives the Born rule as the ratio of the measure of favorable states of the system to the measure of its total possible states.
\end{theorem}
\begin{proof}
Let $\ket{\psi}\in\hilbert$ be a unit vector. The wavefunction $\psi$ defined by $\psi(\phi)=\braket{\phi}{\psi}$ is $\mu$-measurable, so it has measurable support $D\subseteq\mc{C}$. Since the function $\varrho:\mc{C}\to\R$, $\varrho(\phi):=\psi^\ast(\phi)\psi(\phi)$ is measurable and positive, it defines a measure $\mu':=\varrho\mu$ on $D$, and since $\braket{\psi}{\psi}=1$, $\mu'(D)=1$. For any $n\in\N$, we choose a partition $\mc{D}_n=\{D_{n,1},\ldots,D_{n,2^n}\}$ of $D$ into $2^n$ measurable sets of equal measure $\mu'(D_{n,k})=\mu'(D)/2^n=1/2^n$ for all $k\in\{1,\ldots,2^n\}$, so that $\mc{D}_{n+1}$ is a \emph{refinement} of $\mc{D}_n$ (that is, every set from $\mc{D}_n$ is the union of some sets from $\mc{D}_{n+1}$). 
Then, for any $n$, the unit vectors
\begin{equation}
\label{eq:finite_case_psi_partition}
\ket{n,k}:=\sqrt{2^n}\int_{D_{n,k}}\psi(\phi)d\mu(\phi),
\end{equation}
where $k\in\{1,\ldots,2^n\}$, are mutually orthogonal, and
\begin{equation}
\label{eq:finite_case_psi_partition_vectors}
\ket{\psi}=\frac{1}{\sqrt{2^n}}\sum_{k=1}^{2^n}\ket{n,k}.
\end{equation}

The first thing to notice is that we can refine $\ket{\psi}$ from eq. \eqref{eq:finite_case_psi_partition_vectors} as much as we want. This is not possible in the finite-dimensional case, because $2^n$ cannot exceed the dimension of the Hilbert space.

We define, for each $n$ and $\alpha$, the following set indexing the sets consistent with $\mc{C}_{\alpha}$ from each partition $\mc{D}_n$,
\begin{equation}
\label{eq:partitions-invalid}
M_{n,\alpha}:=\{k\in\{1,\ldots,2^n\}|\mu'(D_{n,k}\setminus\mc{C}_{\alpha})=0\}.
\end{equation}

In particular, $\mu'(D_{n,k}\setminus\mc{C}_{\alpha})=0$ is ensured when $D_{n,k}\subseteq\mc{C}_{\alpha}$; in general, $D_{n,k}$ should be ``almost'' included in $\mc{C}_{\alpha}$, except for a zero-measure remaining subset of $D_{n,k}$.

The partitions $\mc{D}_n$ can be chosen so that, as $n\to\infty$, the sets consistent with $\mc{C}_{\alpha}$ approximate $\mc{C}_{\alpha}\cap D$ in the measure $\mu'$, \ie for any $\alpha\in\mc{A}$,
\begin{equation}
\label{eq:partitions-invalid-negligible}
\lim_{n\to\infty}\sum_{k\in M_{n,\alpha}}\mu'(D_{n,k})=\mu'(\mc{C}_{\alpha}\cap D).
\end{equation}

Then, for any $\alpha\in\mc{A}$,
\begin{equation}
\label{eq:finite_case_psi_partition_converge}
\lim_{n\to\infty}\frac{1}{\sqrt{2^n}}\sum_{k\in M_{n,\alpha}}\ket{n,k}=\obs{P}_{\alpha}\ket{\psi}.
\end{equation}

Each partition $\mc{D}_n$ is equivalent to a set of projectors $\mc{P}_n$ that form a partition of the identity and are compatible with the ontic basis, so that each $\mc{P}_{n+1}$ refines $\mc{P}_n$, $\mc{P}_n$ determines a decomposition \eqref{eq:finite_case_psi_partition}, and $\mc{P}_n$ converges to a refinement of $(\obs{P}_{\alpha})_{\alpha\in\mc{A}}$.

From equations \eqref{eq:partitions-invalid-negligible} and \eqref{eq:finite_case_psi_partition_converge} it follows that counting the unit vectors consistent with each macro-projector converges to the Born rule.
Therefore, in the continuous limit, the Born rule results from ``counting'' the basis vectors $(\ket{\phi})_{\phi\in\mc{C}}$, as the ratio of the measure of favorable states of the system to the measure of its total possible states, where the density is $\psi^\ast(\phi)\psi(\phi)d\mu(\phi)$.
\end{proof}

\begin{observation}[Probability]
\label{obs:probability}
If Principle \ref{pp:correspondence} is assumed in quantum mechanics, Theorem \ref{thm:born_rule_counting} shows that the density of the ontic states satisfies the Born rule for the macro-observables $\obs{P}_{\alpha}$, $\alpha\in\mc{A}$.
This allows Condition \ref{cond:correspondence} to be satisfied, and by Observation \ref{obs:equivalence}, Condition \ref{cond:probability} is satisfied too, despite Difficulties \ref{difficulty:parallel}--\ref{difficulty:measure}, according to the Born rule.
\end{observation}

For any physically realistic quantum measurement there is a continuous basis in which the observable is diagonal, as required by Theorem \ref{thm:born_rule_counting}. Even for a single particle in nonrelativistic quantum mechanics, the Hilbert space is infinite-dimensional, and admits continuous bases, \eg the position basis.

\begin{observation}
\label{obs:applicability}
All measurements satisfy, in practice, the continuity condition required for Theorem \ref{thm:born_rule_counting}, because the actual observation is that of the pointer.
\end{observation}

The fact that all measurements are ultimately measurements of positions is advocated for a long time among supporters of Bohm's interpretation, see for example Bell {\cite{Bell2004OnTheProblemOfHiddenVariablesInQuantumMechanics}}, page 10,

\begin{quote}
``Other measurements are reduced ultimately to position measurements. For example, measurement of a spin component means observing whether the particle emerges with an upward or downward deflection from a Stern-Gerlach magnet.''
\end{quote}

The following example details this.
\begin{example}
\label{example:spin}
Consider a measurement of the spin of a particle, whose spin state is initially $\ket{\psi}_{s}=a\ket{\uparrow}_z+b\ket{\downarrow}_z$, where $\abs{a}^2+\abs{b}^2=1$.
The particle also has position degrees of freedom, so its state vector is in fact
\begin{equation}
\label{eq:spin-position}
\psi(\x,t)=\psi_u(\x,t)\ket{\uparrow}_z+\psi_d(\x,t)\ket{\downarrow}_z.
\end{equation}

At the initial time $t_0$, $\psi_u(\x,t_0)=a\psi_0(\x)$, $\psi_d(\x,t_0)=b\psi_0(\x)$, and $\braket{\psi_0}{\psi_0}=1$.
The measurement process consists of using a magnetic field to entangle the spin and the position of the particle, then the position where the deflected particle hits a screen or photographic plate is observed. After passing through the magnetic field, at time $t_1$, $\psi_u(\x,t_1)$ becomes restricted to a region $U$ of the screen, and $\psi_d(\x,t_1)$ to a region $D$. From the resulting position, the spin is inferred to be either ``up'' or ``down''.
The regions $U$ and $D$ of the screen are almost identical, but the norms of $\psi_u(\x,t_1)$ and $\psi_d(\x,t_1)$ are proportional to $\abs{a}^2$ and respectively $\abs{b}^2$.
We obtain, for the two possible regions $U$ and $D$,
\begin{equation}
\label{eq:psi_uniform_spin}
\begin{cases}
\ket{\psi_u,t_1}\ket{\uparrow}_z&=\int_{U} \abs{\psi_u(\x,t_1)} e^{i \theta_u(\x)}\ket{\x} d\x \\
\ket{\psi_d,t_1}\ket{\downarrow}_z&=\int_{D} \abs{\psi_d(\x,t_1)} e^{i \theta_d(\x)}\ket{\x} d\x. \\
\end{cases}
\end{equation}

Then, we invoke collapse or decoherence to explain why only one result is observed. 
This illustrates how, despite apparently making a binary measurement of a qubit, the actual basis is continuous, as required by Theorem \ref{thm:born_rule_counting}.
\qed
\end{example}

\begin{observation}
\label{obs:pp:correspondence-necessity}
In the proof of Theorem \ref{thm:born_rule_counting}, the projectors used to construct the unit vectors \eqref{eq:finite_case_psi_partition} are defined to be compatible with the ontic basis, in order to satisfy Principle \ref{pp:correspondence}. Principle \ref{pp:correspondence} is essential for Theorem \ref{thm:born_rule_counting}.

Any argument based on counting worlds or branch refinements while using the self-locations of the agents should count only the worlds defined by a fixed basis, as Principle \ref{pp:correspondence} states.
Otherwise, we would be forced to admit that any unit vector from any macrostate supports agents, and to count all such vectors as worlds.
This overcounting is inconsistent with the Born rule, as Proposition \ref{thm:overcoungting} will show.
\end{observation}

\begin{proposition}
\label{thm:overcoungting}
If in the proof of Theorem \ref{thm:born_rule_counting} all states from a macrostate are counted, the result contradicts the Born Rule.
\end{proposition}
\begin{proof}
For every $n$, let $\alpha\in\mc{A}$ be so that $\obs{P}_{\alpha}\ket{\psi}\neq 0$. Let $\obs{P}_{\alpha}\ket{\psi}=1/\sqrt{n}\sum_{k=1}^n\ket{\alpha,k}$ be an expansion of $\obs{P}_{\alpha}\ket{\psi}$ as a sum of equal-norm orthogonal vectors from $\hilbert_{\alpha}$, as in eq. \eqref{eq:finite_case_psi}, where $\hilbert_{\alpha}:=\obs{P}_{\alpha}\hilbert$.
We define $N_{n,\alpha}=\{k\in\{1,\ldots,n\}|\ket{n,k}\in\hilbert_{\alpha}\}$, and require that, as $n\to\infty$, $N_{n,\alpha}/n$ converges to $\bra{\Psi}\obs{P}_{\alpha}\ket{\Psi}$.

Then, any unitary transformation of $\hilbert_{\alpha}$ that preserves $\obs{P}_{\alpha}\ket{\psi}$ should transform any of these vectors into another one from an equally valid expansion, and the world supported by the resulting vector should be valid as well.
Denote by $S_{n,\alpha}$ the set of vectors obtained from $\ket{n,k}$ by all unitary transformations of $\hilbert_{\alpha}$ that preserve $\obs{P}_{\alpha}\ket{\Psi}$.
Denote by $\mf{p}(\mc{S})$ be the probability measure of a set $\mc{S}\subseteq\hilbert$ of state vectors counting as worlds.

Let $\alpha\neq\beta\in\mc{A}$ so that $|\obs{P}_{\alpha}\ket{\Psi}|=|\obs{P}_{\beta}\ket{\Psi}|\neq 0$.
Unitary symmetry ensures the existence of a unitary transformation $\obs{S}$ that maps the line $\C\obs{P}_{\beta}\ket{\Psi}\subset\hilbert_{\beta}$ to the line $\C\obs{P}_{\alpha}\ket{\Psi}\subset\hilbert_{\alpha}$, so that one of the following is true: $\obs{S}\hilbert_{\beta}=\hilbert_{\alpha}$, or $\obs{S}\hilbert_{\beta}\subsetneq\hilbert_{\alpha}$, or $\hilbert_{\alpha}\subsetneq\obs{S}\hilbert_{\beta}$.
Unitary symmetry that preserve the macrostates requires that $\mf{p}(\obs{S}\hilbert_{\beta})=\mf{p}(\hilbert_{\alpha})$.
By unitary symmetry, there are infinitely many transformations with the same properties. Let $\obs{S}'$ be another one, so that $\obs{S}'\hilbert_{\beta}\neq\obs{S}\hilbert_{\beta}$. Then, $\obs{S}\hilbert_{\beta}\cap\obs{S}'\hilbert_{\beta}$ is a strict subspace of $\obs{S}\hilbert_{\beta}$, hence $\mf{p}(\obs{S}\hilbert_{\beta}\cap\obs{S}'\hilbert_{\beta})=0$. It follows that $\mf{p}(\obs{S}'\hilbert_{\beta})=\mf{p}(\obs{S}'\hilbert_{\beta}\setminus\obs{S}\hilbert_{\beta})=\mf{p}(\obs{S}\hilbert_{\beta}\setminus\obs{S}'\hilbert_{\beta})=\mf{p}(\obs{S}\hilbert_{\beta})$. From this, $\mf{p}(\hilbert_{\alpha})>\mf{p}(\obs{S}\hilbert_{\beta})+\mf{p}(\obs{S}'\hilbert_{\beta})>\mf{p}(\obs{S}\hilbert_{\beta})=\mf{p}(\hilbert_{\beta})$.
But the Born rule requires that $\mf{p}(\hilbert_{\alpha})=\mf{p}(\hilbert_{\beta})$, which is satisfied only if $\obs{S}\hilbert_{\beta}=\hilbert_{\alpha}$ for every $\alpha\neq\beta\in\mc{A}$.

But if $\obs{S}\hilbert_{\beta}=\hilbert_{\alpha}$ for every $\alpha\neq\beta\in\mc{A}$, it can be shown that this contradicts the Born Rule for $|\obs{P}_{\alpha}\ket{\Psi}|>|\obs{P}_{\beta}\ket{\Psi}|$.
To see this, we need to calculate the angle $\omega_{n,\alpha}$ between $\ket{n,k'}$ and $1/\sqrt{n}\sum_{k\in N_{n,\alpha}}\ket{n,k}$ when $k'\in N_{n,\alpha}$, because it determines $S_{n,\alpha}$. It is given by $\cos\omega_{n,\alpha}=\abs{\bra{n,k'}1/\sqrt{n N_{n,\alpha}}\sum_{k\in N_{n,\alpha}}\ket{n,k}}=1/\sqrt{n N_{n,\alpha}}$.
Then, as $n\to\infty$, $\omega_{n,\alpha}\to \pi/2$, for all $\alpha$.
Therefore, as $n\to\infty$, $\mf{p}(S_{n,\alpha})/\mf{p}(S_{n,\beta})=1$.
So counting all vectors from the sets $S_{n,\alpha}$ as worlds contradicts the Born Rule.
\end{proof}

Theorem \ref{thm:born_rule_counting} can be seen as the continuous limit of the refined branch-counting method proposed in \cite{Saunders2021BranchCountingInTheEverettInterpretationOfQuantumMechanics}, amended with Principle \ref{pp:correspondence}. This amendment is demanded by Proposition \ref{thm:overcoungting}.

\section{Interpretation of the wavefunction}
\label{s:interpretation}

\subsection{Uniformization of the ontic states}
\label{s:density}

The continuous limit from the proof of Theorem \ref{thm:born_rule_counting} keeps the unit vectors orthogonal, so that we can interpret any state vector as the limit of a sum of the form \eqref{eq:finite_case_psi}, and by this, the basis vectors as distinct outcomes in the sample space (\cf Definition \ref{def:probability-space}).

Without loss of generality, for any given state vector $\ket{\psi}$ so that $\abs{\psi(\phi)}$ is a $\mu$-measurable function of $\phi$, we can assume that $\psi(\phi)\in\R$ for all $\phi$.
If not, substitute the basis by $\ket{\phi}\mapsto e^{i \theta(\phi)}\ket{\phi}$, where $\theta(\phi)$ is the phase appearing in the polar form of $\braket{\phi}{\psi}$, for all $\phi\in\mc{C}$.

\begin{proposition}
\label{thm:born_rule_counting_measure}
The state vector $\ket{\psi}$ has the form
\begin{equation}
\label{eq:psi_uniform}
\ket{\psi}=\int_{\mc{C}}\ket{\phi} d\wt{\mu}(\phi),
\end{equation}
where $\theta:\mc{C}\to\R$, and $\wt{\mu}$ is a measure on $\mc{C}$.

Any projector $\obs{P}_{\alpha}$ diagonal in the basis $(\ket{\phi})_{\phi\in\mc{C}}$ corresponds to a subset $\mc{C}_{\alpha}\subseteq \mc{C}$.
If $\mc{C}_{\alpha}$ is $\mu$-measurable,
\begin{equation}
\label{eq:born_rule_continuous}
\left|\int_{\mc{C}_{\alpha}}\ket{\phi} d\wt{\mu}(\phi)\right|^2=\int_{\mc{C}_{\alpha}}r^2(\phi)d\mu(\phi).
\end{equation}
\end{proposition}
\begin{proof}
Let $r(\phi):=\abs{\braket{\phi}{\psi}}$. Then, $r\in L^2(\mc{C},\mu,\R)$ is a real non-negative square-integrable function, and
\begin{equation}
\label{eq:psi_real}
\ket{\psi}=\int_{\mc{C}}r(\phi)\ket{\phi} d\mu(\phi).
\end{equation}

The following measure satisfies eq. \eqref{eq:psi_uniform},
\begin{equation}
\label{eq:nonuniform-measure}
d\wt{\mu}(\phi):=r(\phi)d\mu(\phi).
\end{equation}

Then, evidently $\obs{P}_{\alpha}\ket{\psi}=\int_{\mc{C}_{\alpha}}\ket{\phi} d\wt{\mu}(\phi)$, and

\begin{equation}
\label{eq:normalized_proof}
\left|\int_{\mc{C}_{\alpha}}\ket{\phi} d\wt{\mu}(\phi)\right|^2=\bra{\psi}\obs{P}_{\alpha}\ket{\psi}=\int_{\mc{C}_{\alpha}}r^2(\phi)d\mu(\phi). 
\end{equation}

It may seem that the vector $\int_{\mc{C}_{\alpha}}\ket{\phi} d\wt{\mu}(\phi)$ cannot have finite norm, or that its norm may be larger than $1$.
So let us check this more explicitly:
\begin{equation}
\label{eq:check-normalized}
\begin{aligned}
\left|\int_{\mc{C}_{\alpha}}\ket{\phi} d\wt{\mu}(\phi)\right|^2
&=\(\int_{\mc{C}_{\alpha}}\bra{\phi} d\wt{\mu}(\phi)\)\(\int_{\mc{C}_{\alpha}}\ket{\phi'} d\wt{\mu}(\phi')\)\\
&=\int_{\mc{C}_{\alpha}}\(\int_{\mc{C}_{\alpha}}\braket{\phi}{\phi'} d\wt{\mu}(\phi')\)d\wt{\mu}(\phi)\\
&=\int_{\mc{C}_{\alpha}}\(\int_{\mc{C}_{\alpha}}\braket{\phi}{\phi'} r(\phi')d\mu(\phi')\)d\wt{\mu}(\phi)\\
&=\int_{\mc{C}_{\alpha}}r(\phi)d\wt{\mu}(\phi)
=\int_{\mc{C}_{\alpha}}r^2(\phi)d\mu(\phi).\\
\end{aligned}
\end{equation}

This double-checks eq. \eqref{eq:born_rule_continuous}.
\end{proof}

Figure \ref{born.pdf} illustrates the uniform-amplitude representation of the wavefunction from eq. \eqref{eq:psi_uniform}.

\image{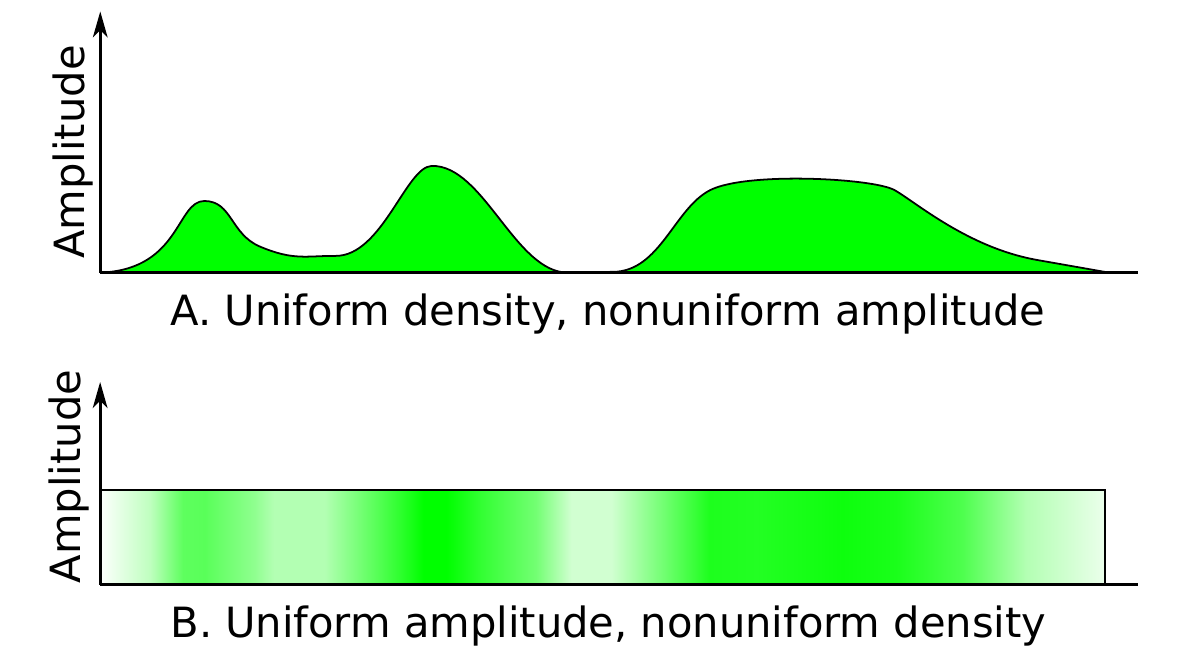}{1}{\textbf{The Born rule from ``counting'' basis states.}
\\\textbf{A.} The usual interpretation of a wavefunction as a linear combination of basis state vectors of different amplitudes.
\\\textbf{B.} The interpretation of the wavefunction in terms of basis vectors representing ontic states.}

\begin{observation}
\label{obs:true-density}
We notice the existence of three densities. The first one is $d\mu$, given by the measure on $\mc{C}$, and it is independent of states.
The second density is $d\wt{\mu}=r(\phi)d\mu$, which describes how the ontic states contribute to the state vector $\ket{\psi}$ in eq. \eqref{eq:psi_uniform}.
The third density, $d\mu'(\phi)=r^2(\phi)d\mu$ is the probability density corresponding to the Born rule, as in eq. \eqref{eq:born_rule_continuous}.
\end{observation}

This may seem strange, despite the explicit calculation from \eqref{eq:check-normalized}, so let us try to understand the interplay between these densities.

\begin{remark}[``Magic'' accident]
\label{remark:squared_amplitude}
One may expect that we have to define the measure $\wt{\mu}$ so that $d\wt{\mu}(\phi)$ is $r^2(\phi)d\mu(\phi)$, rather than as in eq. \eqref{eq:nonuniform-measure}.
But, interestingly, eq. \eqref{eq:born_rule_continuous} follows without this, simply by choosing the measure $\wt{\mu}$ so that the amplitudes become uniformly equal to $1$.
Moreover, it does not even work otherwise, because $\ket{\psi}\neq\int_{\mc{C}}r^2(\phi)\ket{\phi} d\mu(\phi)$.
\qed
\end{remark}

\begin{remark}[Why does it work?]
\label{remark:uniform}
Naively, it may seem that the norm of $\int_{\mc{C}_{\alpha}}\ket{\phi} d\wt{\mu}(\phi)$ cannot be finite, or at least that it is equal to $\int_{\mc{C}_{\alpha}}d\wt{\mu}(\phi)$ and it can be larger than $1$, but this is incorrect.
Eq. \eqref{eq:born_rule_continuous} is correct, as checked in \eqref{eq:normalized_proof} and double-checked in \eqref{eq:check-normalized}, because $r(\phi)$ is square-integrable, and since it is $\mu$-measurable, the measure $\wt{\mu}$ is \emph{absolutely continuous} with respect to $\mu$.

There is a reason why, in eq. \eqref{eq:check-normalized},
\begin{equation}
\label{eq:int_r_phi}
\int_{\mc{C}_{\alpha}}\braket{\phi}{\phi'} d\wt{\mu}(\phi')=r(\phi)
\end{equation}
rather than $1$.
A perhaps more revealing way of understanding this involves the \emph{scaling property} of the Dirac distribution $\delta(x)$ with $a>1$,
\begin{equation}
\label{eq:dirac_scaling}
\delta(a x)=a^{-1}\delta(x).
\end{equation} 

To see how this works, consider the Hilbert space  $L^2(\R^\n,\mu,\C)$ with the basis $\(\ket{\x}\)_{\x\in\R^\n}$.
If $\mathbf{f}:\R^\n\to\R^\n$ is an invertible reparametrization of $\R^\n$, by making a change of variables $\wt{\y}=\mathbf{f}(\y)$ we obtain the following generalization of eq. \eqref{eq:dirac_scaling},
\begin{equation}
\label{eq:dirac_scaling_generalized}
\int_{\R^\n}\braket{\x}{\y}\de \wt{\y}
= \int_{\R^\n}\braket{\x}{\y}\abs{\frac{\partial\mathbf{f}}{\partial\y}}\de\y
=\abs{\frac{\partial\mathbf{f}}{\partial\x}},
\end{equation}
where $\abs{\partial\mathbf{f}/\partial\x}$ is the modulus of the determinant of the Jacobian matrix of $\mathbf{f}$ at $\x$.
With the notation from Proposition \ref{thm:born_rule_counting_measure} but $\phi$ replaced by $\x$, $\mu$ is the Lebesgue measure on $\R^\n$, $d\mu(\x)=\de\x$, $d\wt{\mu}(\y)=\de\wt{\y}$, and
\begin{equation}
\label{eq:dirac_scaling_generalized_r}
r(\x)=\frac{d\wt{\mu}(\x)}{d\mu(\x)}=\abs{\frac{\partial\mathbf{f}}{\partial\x}}.
\end{equation}

This explains once more how the homogenization of the amplitude from eq. \eqref{eq:psi_uniform}, despite not involving $r^2(\phi)$, leads to its appearance in eq. \eqref{eq:check-normalized}, by using the generalized scaling property of the Dirac delta distribution.
\qed
\end{remark}

\begin{observation}
\label{obs:interpretation-r}
This explains that, while the probability density of the ontic states $\phi$ is $r^2(\phi)d\mu$, the state vector $\ket{\psi}$ taking part in the {\schrod} equation is composed by using the density $d\wt{\mu}$, as in eq. \eqref{eq:psi_uniform}.
In fact, this is the continuous limit of eq. \eqref{eq:finite_case_psi_partition_vectors}.
\end{observation}

\subsection{Wavefunction or wavefunctional?}
\label{s:wavefunctional}

Subsystems admit observables that cannot be diagonalized simultaneously, so their continuous bases depend on the observable. 
But since different measurement settings ultimately translate into distinguishing macrostates defined by the same set of macro projectors, the ontic basis from Observation \ref{obs:ontic-basis} and Principle \ref{pp:correspondence} is consistent with any observables we measure for the subsystems \cite{Stoica2022ActualQuantumObservablesAreCompatible}.
This universal basis can be taken as representing ``classical states'', which may be called \emph{ontic states}. Theorem \ref{thm:born_rule_counting} allows us to interpret the Born rule for any measurement as ``counting'' such ontic states.

But what are these ontic states?
Since each particle is represented on a Hilbert space of wavefunctions that have, among their degrees of freedom, the positions, which play a role in any measurement, and also form a continuous basis, it may be tempting to interpret the ontic states as position eigenstates, as in Example \ref{example:spin}. But we know that in fact the universe is not described by nonrelativistic quantum mechanics, but by quantum field theory, in which there are no localized particles.

A unique basis $(\ket{\phi})_{\phi\in\mc{C}}$ that really is ontic or classical is possible in quantum field theory.
In the {\schrod} wavefunctional formulation of quantum field theory \cite{Jackiw1988AnalysisInfDimManifoldsSchrodingerRepresentationForQuantizedFields,Hatfield2018QuantumFieldTheoryOfPointParticlesAndStrings}, $\mc{C}$ becomes the configuration space of classical fields, and the {\schrod} \emph{wavefunctional}
\begin{equation}
\label{eq:wavefunctional}
\Psi[\phi]:=\braket{\phi}{\Psi}
\end{equation}
replaces the nonrelativistic wavefunction.
Here, $\phi$ stands for a collection of classical fields, $\phi=(\phi_1,\ldots,\phi_n)$.
The configuration space $\mc{C}$ is endowed with a measure $\mu$\footnote{To admit a Lebesgue measure $\mu$, the classical configuration space $\mc{C}$ should be finite-dimensional.
Maybe it is, because the \emph{entropy bound} \cite{Bekenstein1981UniversalUpperBound,Bekenstein2005HowDoesEntropyInformationBoundWork} requires the Hilbert space of fields defined on compact regions of space to have finite dimension. Also the fields are constrained by equations, the gauge degrees of freedom have to be factored out, and there are severe constraints related to the arrow of time \cite{Stoica2024DoesQuantumMechanicsRequireConspiracy}. So we assume that $\mc{C}$ admits a measure $\mu$.}.

\subsection{Macro-classicality}
\label{s:macro-classicality}

The wavefunctional formulation represents quantum states in terms of classical field states, in the sense that the wavefunctional is a complex functional defined on the configuration space of classical fields. The usual Fock representation can be obtained from the basis $(\ket{\phi})_{\phi\in\mc{C}}$ \cite{Hatfield2018QuantumFieldTheoryOfPointParticlesAndStrings}. The Fock representation can then be used to interpret the quantum fields in terms of more commonly used nonrelativistic quantum mechanical wavefunctions and operators. But this is a departure from the more foundational description provided by wavefunctionals.

We never observe individual particles directly, but only macrostates. Macrostates are imported from the classical theory, and they are empirically adequate, because at the macro level the universe looks classical. 
Therefore, Principle \ref{pp:correspondence}, which says that states of the form $\ket{\phi}$ belong to macrostates, \ie for every $\ket{\phi}$ there is a macrostate $\obs{P}_{\alpha}\hilbert$ so that $\ket{\phi}\in\obs{P}_{\alpha}\hilbert$, makes sense.
\begin{principle}
\label{pp:micro-macro}
At any instant, at the macro level, a classical universe in the classical state $\phi$ looks the same as a quantum universe in the quantum state $\ket{\phi}$ or linear combinations of such states from the same macrostate.
\end{principle}

And indeed, it took us a very long time to realize that our universe is not classical, but quantum.

\subsection{Interpretation of complex numbers}
\label{s:complex}

Recall that eq. \eqref{eq:psi_uniform} is based on absorbing the phase factor in the vector by substituting $\ket{\phi}\mapsto e^{i \theta[\phi]}\ket{\phi}$,  done just before stating Proposition \ref{thm:born_rule_counting_measure}.
This substitution depends on the state $\ket{\Psi}$, in particular $\theta[\phi]$ changes in time.
So we cannot simply interpret $\ket{\Psi}$ directly as a probability density over the classical states.

But the phase change $\ket{\phi}\mapsto e^{i \theta[\phi]}\ket{\phi}$ can be identified with an $\U(1)$ gauge transformation of the classical field, denoted $\phi\mapsto e^{i \theta[\phi]}\phi$ (in fact $\U(1)$ acts differently on different fields, but I will use a uniform notation for its action), so that
\begin{equation}
\label{eq:phase-gauge}
e^{i \theta[\phi]}\ket{\phi}\equiv\ket{e^{i \theta[\phi]}\phi}.
\end{equation}

This makes sense because
(1) multiplying a state vector with a phase factor changes the vector, but not the physical (quantum) state it represents,
and
(2) an $\tn{U}(1)$ gauge transformation of a classical field represents the same physical (classical) state.

Charged and spinor fields, and electromagnetic potentials, admit a nontrivial $\U(1)$ symmetry, but it is sufficient that $\phi$ includes one such field.
The gauge transformation depends on the state $\ket{\Psi}$, so it changes in time.

\begin{observation}
\label{obs:gauge}
$\Psi[\phi]$ \emph{can be made real by changing the global $\U(1)$ gauge of the basis of classical fields}.
\end{observation}

\begin{principle}
\label{pp:wavefunctional}
The wavefunctional $\ket{\Psi}=\int_{\mc{C}}\ket{\phi} d\wt{\mu}[\phi]$ can be interpreted as a set of gauged classical fields distributed according to a density functional (Fig. \ref{psi-interpretation.pdf}).
\end{principle}

\image{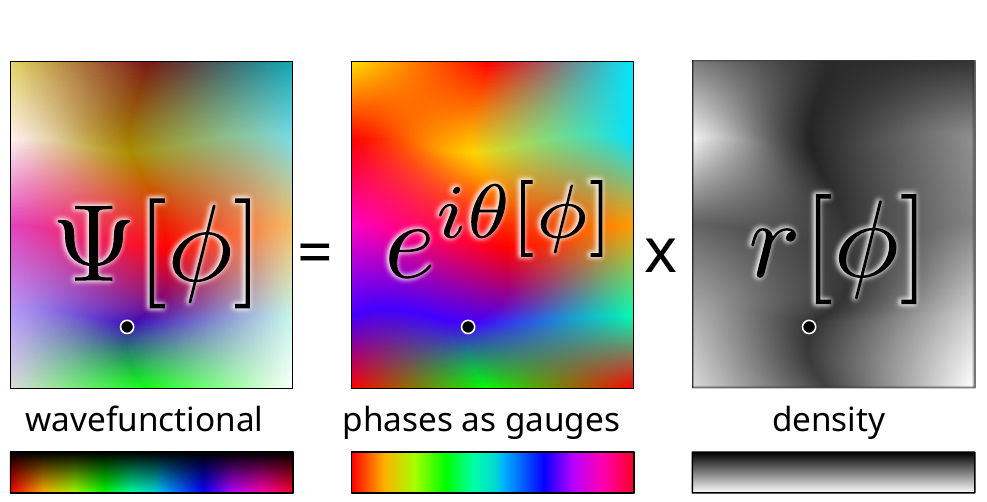}{1}{\textbf{Interpretation of the wavefunctional.} The $\U(1)$ gauge or phase is represented by the pure color hues in the color wheel. The density $r[\phi]$, represented as shades of gray, is the density of state vectors as they combine in $\ket{\Psi}$ (while the probability density over the sample space $\mc{C}$ is $r^2[\phi]$). Their combination gives the wavefunctional $\ket{\Psi}=\int_{\mc{C}}\ket{\phi} d\wt{\mu}[\phi]$ as a set of classical fields with varying densities and gauges.}

In this way, the probability density $r^2[\phi]$ is similar to the probability density in a classical system.
This relation becomes visible for a quantum system obtained as the Koopman-von Neumann representation of a classical system \cite{Koopman1931HamiltonianSystemsAndTransformationInHilbertSpace,vonNeumann1932KoopmanMethod}. In this case, indeed the classical probability density becomes $r^2[\phi]$.

\subsection{Local beables}
\label{s:local-beables}

There are several benefits in using the interpretation of the wavefunctional from Principle \ref{pp:wavefunctional} as starting point in the investigations of the foundations of quantum theory.
It is more foundational, since quantum field theory is more foundational than nonrelativistic quantum mechanics.
It comes with an ontology -- each state $\ket{\phi}$ corresponds to a set of fields defined on the $3$d-space, not on the configuration space. These fields are the \emph{local beables}. The necessity of local beables was extensively advocated by Bell \cite{Bell2004SpeakableUnspeakable}.
The Born rule can be interpreted in terms of such ontic states, based on Principle \ref{pp:correspondence}.

A state does not consist of a single ontic state, but of a set of such states (Principle \ref{pp:wavefunctional}).
The Projection Postulate should not be understood as collapsing the system to a basis state $\ket{\phi}$, no measurement can extract the complete information about the state of the entire universe.
Only the ontic states making $\Psi[\phi]$ belonging to the resulting macrostate $\obs{P}_{\alpha}\hilbert$ should remain after the projection.

Based on the above discussion, let us try to build a version of standard quantum mechanics based on Principles \ref{pp:correspondence}, \ref{pp:micro-macro}, \ref{pp:wavefunctional}, and the Projection Postulate. Before the projection occurs, the state is a superposition of ontic states. But even after the projection occurs, the state is still a superposition of ontic states, because the macro-projectors are always compatible with more ontic states.
If only one would remain, a measurement of an observable $\wh{A}$ followed by a measurement of another observable $\wh{B}$ that commutes with $\wh{A}$, followed by a repetition of the measurement of $\wh{A}$, would give the wrong probabilities, because the measurement of $\wh{B}$ would change the value of the observable $\wh{A}$.
It follows that
\begin{observation}
\label{obs:self-location}
Even in a version of standard quantum mechanics based on Principles \ref{pp:correspondence}, \ref{pp:micro-macro}, \ref{pp:wavefunctional}, and the Projection Postulate, the probability density $r^2[\phi]$ is more naturally understood as self-location probability, rather than epistemic probability (see Observation \ref{obs:equivalence}).
\end{observation}

Therefore, we are led naturally to a version of the many-worlds interpretation based on Principles \ref{pp:correspondence}, \ref{pp:micro-macro}, and \ref{pp:wavefunctional}.

\subsection{Many-worlds}
\label{s:MWI}

If decoherence makes the components of $\Psi[\phi]$ corresponding to different macrostates no longer interfere, there is no need to invoke the Projection Postulate, and we can adopt the \emph{many-worlds interpretation} (MWI), based on Principles \ref{pp:correspondence}, \ref{pp:micro-macro}, and \ref{pp:wavefunctional}.

Observation \ref{obs:self-location} shows that, once we adopt Principles \ref{pp:correspondence}, \ref{pp:micro-macro}, and \ref{pp:wavefunctional}, we should already assume many-worlds even in standard quantum mechanics, particularly if we want it to satisfy Condition \ref{cond:probability} and to be able to support the derivation of the Born rule from Theorem \ref{thm:born_rule_counting} interpreted according to Principle \ref{pp:correspondence}.

\begin{observation}
\label{obs:MWI-counting}
``Counting'' micro-branches that correspond to the basis $(\ket{\phi})_{\phi\in\mc{C}}$ gives the correct probabilities in the MWI, in agreement with Condition \ref{cond:probability}.
Even if, unlike the macro-branches, the micro-branches may interfere in the future, they interfere within the same macro-branch.
Moreover, since each micro-branch consists of classical fields $\phi$, and since these are the local beables, it becomes justified to count each micro-branch as a world, as stated by Principle \ref{pp:correspondence}.
\end{observation}

\begin{observation}
\label{obs:background-freedom}
We should also include quantum gravity in our foundational investigations of quantum theory. In background-free approaches to quantum gravity, it becomes impossible to physically interpret all linear combinations as superpositions, because states in which the geometry of space is different cannot be superposed unambiguously, so the ontic states dissociate automatically \cite{Stoica2022BackgroundFreedomLeadsToManyWorldsLocalBeablesProbabilities}. They can reassociate, unless the dissociation becomes irreversible due to decoherence. This provides an additional justification for the many-worlds interpretation (in the revised form from \cite{Stoica2022BackgroundFreedomLeadsToManyWorldsLocalBeablesProbabilities}).
\end{observation}

\section{Discussion}
\label{s:discussion}

The article presents a way to understand quantum mechanics that makes sense of the probabilities in a way close to that in classical physics, perhaps better, as I shall argue below.
The starting point is Principle \ref{pp:correspondence}, which insists on the existence of an ontic basis, which should be understood as the existence of a distribution of classical worlds. Principle \ref{pp:micro-macro} connects the microstates, which are instantaneous classical worlds, to the macrostates. This does not preclude the quantum character of the world, because at the macro level only the macro-classicality is perceived, while micro-classicality is hidden, and the dynamics ensures the quantum behavior noticed through measurements. Since the dynamics do not maintain a continuity of the microstates, such a classical world does not maintain its identity across time, but it evolves into more classical worlds. The dynamics make them manifests interference, and this is why quantum phenomena exist just like in standard quantum mechanics.
If the wavefunction would be real rather than complex, it would be equivalent with a probability distribution of coexisting classical worlds. But since the complex phase is present and the evolution of the density depends on it, we need a way to integrate phase into the classical worlds. This can be done naturally if the classical worlds would be classical field configurations, in the wavefunctional formulation of quantum mechanics. In this case, a classical field configuration can have associated a phase that is, classically, the gauge, and transforms under the gauge symmetries. Principle \ref{pp:wavefunctional} stipulates this correspondence.

Probabilities are as in classical physics, with the difference that the probability distribution is not about a single classical world whose microstate is incompletely known, but about many parallel classical worlds in which the agent can be instantiated.

As mentioned in Remark \ref{rem:classical-probability}, the classical interpretation of  probabilities is problematic even in a classical world, mainly due to the appeal to the Principle of Indifference. Alternative interpretations of probabilities exist, but problems still persist. The frequentist approach works for a too large number of repetitions to extract a statistical significant probability, and this leaves probabilities for unique or rare events uncovered. The Bayesian approach is perhaps the most promising, and in emphasizing the necessity to the priors, it leads to the understanding that probabilities in the physical world should rely be conditioned by an assumption about the very initial state of the universe \cite{Boltzmann1964LecturesOnGasTheory}, assumption sometimes called the \emph{Past Hypothesis} \cite{DavidZAlbert2015AfterPhysics,Wallace2012TheEmergentMultiverseQuantumTheoryEverettInterpretation}. This is not a drawback of the Bayesian approach, but a merit, since it reveals this necessity.

The problem with the Principle of Indifference is due, at least partially, to the necessity to appeal to counterfactuals. This is a problem: if only one world exists, the probability distribution should be a Dirac function centered on the actual world. Any probability arisen from our insufficient knowledge of the microstate should be derived by deduction or induction from the actual world alone, and so it would not be a probability distribution over possible worlds, of which only one really exists. Even with the adoption of a Bayesian position, this undermines the idea of counterfactuals grounding the very idea of epistemic probability, since this is grounded in something that does not exist. I think, along with David Lewis \cite{Lewis1973Counterfactuals,Lewis1986OnThePluralityOfWorlds}, that the counterfactuals should have physical existence as well. But of course, if they are inaccessible to us, the probability distribution itself is inaccessible, hence again the necessity of a Bayesian approach and of the Past Hypothesis. However, since the world is quantum, and since the dynamics does not connect one-to-one various temporal instances of the worlds, but it rather allows them to mix in the way we know as interference, these worlds have effects on each other. This gives a strong justification of taking the parallel worlds seriously as physically real rather than as non-physical possibilities. In my opinion, this shows that ignorance about self-location provides a better grounding for probabilities than ignorance about the microstate. Perhaps even if the world would have been classical, parallel worlds would have had to be invented to ground probabilities physically.

As mentioned in the Introduction, Gleason's theorem does not justify the Born rule in a probabilistic sense, but only as a measure, and he never claimed more. The confusion between measure and probability measure may come from the way mathematicians define and use probability measures as a particular case of a measure. But in the physical domain not any non-negative measure is a probability measure, as we can see from the examples of mass density and other densities encountered in physics. Therefore, to solve this problem, many researchers tried to put on a firm ground a reason to interpret the squared amplitude probabilistically, mainly in the context of MWI.
The most straightforward seems to be to write the state vector as a linear combination of equal-amplitude state vectors corresonding to definite outcomes of the measurements, and to count them \cite{Everett1957RelativeStateFormulationOfQuantumMechanics,deWittGraham1973ManyWorldsInterpretationOfQuantumMechanics,Saunders2024FiniteFrequentismExplainsQuantumProbability}.
While constrained by the outcome definiteness requirement for the worlds in the decomposition, this is arbitrary, since infinitely many different ways to achieve this are possible whenever one way exists.
If we would count all state vectors in all of these decompositions we would obtain an overcounting, as shown in Proposition \ref{thm:overcoungting}.
Deutsch, followed by Wallace, and Saunders, brought various compelling reasons based on decision theory \cite{Deutsch1999QuantumTheoryOfProbabilityAndDecision,Wallace2002QuantumProbabilitiesAndDecisionRevisited,Saunders2004BornRuleFromOperationalAssumptions}. Another motivation was based on the ``measure of existence'', introduced by Vaidman \cite{Vaidman2012ProbabilityInMWI}.

I will not dispute these justifications of the squared amplitudes as probabilities, they are compelling, and since the amplitudes are forced onto us as providing probabilities if we are not to violate the Schr\"odinger unitary evolution, there must be a reason why this should work.
With the proposal from this article, I just want to suggest that maybe the squared amplitudes speak about the classical worlds as explained here, and this is what grounds their understanding as probabilities in both the decision-theoretic and the ``measure of existence'' proposals.

The proposal of a different version of MWI from Section \sref{s:MWI} is not the first attempt to make sense of the quantum probabilities as distributions of classical worlds. 
The suggestion seems to be already built in Bohm's theory, which takes the world as consisting of point-particles guided by the wavefunction \cite{Bohm1952SuggestedInterpretationOfQuantumMechanicsInTermsOfHiddenVariables}. The idea that all possible configurations of point-particles may coexist distributed according to the wavefunction's squared amplitude was put forward by Bell in 1976 \cite{Bell2004TheMeasurementTheoryOfEverettAndDeBrogliesPilotWave}, although he argued against Everett's interpretation and in favor of Bohm's.
A similar proposal, but more developed, was made by Tipler \cite{Tipler2006WhatAboutQuantumTheoryBayesAndTheBornInterpretation}.

More recently, Bostr\"om as well proposed to keep all Bohmian configurations and the wavefunction to guide them as in Bohm's theory, similar to Bell \cite{Bostrom2015QuantumMechanicsAsADeterministicTheoryOfAContinuumOfWorlds}.

Interestingly, it is possible to ``absorb'' the information from the wavefunction by rewriting the guiding equation in terms of the higher order derivatives, as noticed by Poirier and Schiff \cite{Poirier2010BohmianMechanicsWithoutPilotWaves,SchiffPoirier2012CommunicationQuantumMechanicsWithoutWavefunctions}, and independently at the same time by Raykin \cite{Raykin2012AnalyticalQuantumDynamicsInInfinitePhaseSpace}. This allows for a single world as in Bohm's theory, but without the wavefunction. However, the absorption of the wavefunction in the higher-order derivatives does not seem to be physically different from fully including the wavefunction in the dynamical law, as proposed by D{\"u}rr, Goldstein, and Zangh{\`i} \cite{Durr1995BohmianWavefunction}. But Poirier and Schiff tried to replace the wavefunction with a probability distribution over point-particle configurations. By contrast, Bostr\"om found it necessary to keep the wavefunction as well. After all, the phases cannot be absorbed in the point-particle configurations, and they play an active role in the dynamics.

The proposal of Hall, Deckert, and Wiseman \cite{HallDeckertWiseman2014QuantumPhenomenaModeledByInteractionsBetweenManyClassicalWorlds}, based on taking the continuum limit of a finite number of classical interacting worlds, attempts to remove the need for the wavewfunction by recovering the phase from the velocities. Other proposals based on classical configurations include Sebens' ``many interacting worlds'' \cite{Sebens2015QuantumMechanicsAsClassicalPhysics}, Tappenden's version of MWI \cite{Tappenden2023SetTheoryAndManyWorlds}, and, although it does not assume classical trajectories but only instantaneous configurations, Arve's version of MWI \cite{Arve2020EverettsMissingPostulateAndTheBornRule}.

The proposal from Section \sref{s:MWI} differs from these approaches in a couple of aspects. The first difference, which may be inessential, is that in this proposal the classical worlds do not maintain their identity during the evolution. This was already discussed by Bell in 1976 \cite{Bell2004TheMeasurementTheoryOfEverettAndDeBrogliesPilotWave}, page 98,
\begin{quote}
``But we learn from Everett that if we do not like these trajectories we can simply leave them out. We could just as well redistribute the configuration {$(\mathbf{x}_1, \mathbf{x}_2,\ldots)$} at random (with weight {$\left\lvert\psi\right\rvert^2$}) from one instant to the next. For we have no access to the past, but only to memories, and these memories are just part of the instantaneous configuration of the world. [... This] is logically coherent, and does not need to supplement mathematical equations with vague recipes. But I do not like it. Emotionally, I would like to take more seriously the past of the world (and of myself) than this would permit.''
\end{quote}

The proposals of Poirier \& Schiff and Bostr\"om ensure this continuity. Mine does not, at least not if it remains in the form stated here. But I do not think it is a problem, and if it is, we immediately can add a continuity of the worlds across time, based on the guiding equation like these other authors did, but for the wavefunctional. A way to do this is by using Valentini's guiding equation for field configurations \cite{Valentini1996PilotWaveTheoryOfFieldsGravitationAndCosmology}. However, if we would have two dynamical equations, or pay the price of having both the wavefunction and the classical worlds, and a potential connecting the classical worlds as arguments of the guiding equation, wouldn't it be much cheaper to stick with Bohm's theory, which has two dynamical equations and he wavefunction, but only a classical world? Therefore, while it is possible to restore the time continuity of the classical worlds, I find it too expensive to my taste.

Another difference is that the version of MWI proposed here does not rely on non-relativistic quantum mechanics, but on the quantum field theory, in the wavefunctional formulation. This allows the totality of classical worlds fully recover the wavefunctional.
By using the wavefunctional, the phase of each classical world can be included as a gauge in each of the classical worlds, as seen in Section \sref{s:complex}. Therefore, in this proposal the wavefunctional is fully recoverable from the ensemble of coexisting classical worlds.

In consequence, this proposal can be understood as a natural continuation of classical physics, based on the following steps. First, replace the classical ensemble of worlds, which contains a single unknown world and many other possible but non-physical worlds, with an ensemble of coexisting worlds.
The classical worlds are taken to consist of classical field configurations in the wavefunctional formulation of quantum theory.
The dynamics, given by Schr\"odinger's wavefunctional equation, applies to the full ensemble, not to the individual worlds. 
The complex phase factor in quantum mechanics can be fully replaced by gauges of the classical field configurations. In this way, the full wavefunctional is replaced by an ensemble of classical coexisting and co-evolving worlds.

%
%
%
%
%
%
%
%
%
%
%
%


\end{document}